\documentclass[final,onefignum,onetabnum]{siamart190516}

\usepackage[english]{babel}
\usepackage[utf8]{inputenc}
\usepackage{amsfonts}
\usepackage{amssymb}
\usepackage{mathtools}
\usepackage{tikz}
\usepackage[labelfont=bf]{caption}
\usepackage{subcaption}

\newsiamthm{notation}{Notation}
\newsiamremark{remark}{Remark}
\newsiamthm{claim}{Claim}
\newenvironment{claimproof}{\noindent\textit{Proof.}}{\hfill\ensuremath{\blacksquare}
	\medskip}
\usepackage{etoolbox}
\newcounter{algcounter}
\newenvironment{example}[1][]{
	\bigskip\noindent
	\refstepcounter{algcounter}
	\textbf{Example \thealgcounter}
    \ifx&#1&
	\else
	(#1).
	\fi 
	}{\hfill\ensuremath{\lrcorner}\\\medskip}
\numberwithin{algcounter}{section}
\usepackage{etoolbox}

\DeclareMathOperator{\dist}{d}
\DeclareMathOperator{\mind}{\delta_{\dist}}
\DeclareMathOperator{\minsome}{\delta}
\DeclareMathOperator{\ord}{ord}
\DeclareMathOperator{\At}{At}

\DeclareMathOperator{\Sing}{Sing}
\DeclareMathOperator{\wid}{wd}
\DeclareMathOperator{\bw}{bw}
\DeclareMathOperator{\tn}{tn}
\newcommand{\equivalent}{equivalent}
\newcommand{\mindistftk}{maximum linkage function}
\newcommand{\MinDistFtk}{Maximum Linkage Function}

\title{Tangles and Hierarchical Clustering}

\date{\today}

\author{Eva Fluck
	\thanks{RWTH Aachen, Chair for Computer Sience 7, Logic and Theory of Discrete Systems (\email{fluck@cs.rwth-aachen.de})}
}

\headers{Tangles and Hierarchical Clustering}{Eva Fluck}

\begin{document}

\maketitle

\begin{abstract}
We establish a connection between tangles, a concept from structural graph theory that plays a central role in Robertson and Seymour's graph minor project, and hierarchical clustering.
Tangles cannot only be defined for graphs, but in fact for arbitrary connectivity functions, which are functions defined on the subsets of some finite universe.
In typical clustering applications these universes consist of points in some metric space.
Connectivity functions are usually required to be submodular.
It is our first contribution to show that the central duality theorem connecting tangles with hierarchical decompositions (so-called branch decompositions) also holds if submodularity is replaced by a different property that we call maximum-submodular.
We then define a connectivity function on finite data sets in an arbitrary metric space and prove that its tangles are in one-to-one correspondence with the clusters obtained by applying the well-known single linkage clustering algorithms to the same data set.
Lastly we generalize this correspondence for any hierarchical clustering.
We show that the data structure that represents hierarchical clustering results, called dendograms, are equivalent to maximum-submodular connectivity functions and their tangles.
The idea of viewing tangles as clusters has first been proposed by Diestel and Whittle in 2016 as an approach to image segmentation.
To the best of our knowledge, our result is the first that establishes a precise technical connection between tangles and clusters.
\footnote{A conference version that contains some of the results has appeared at MFCS 2019 \cite{fluck2019}.}
\end{abstract}



\section{Introduction}

Connectivity in graphs and connectivity systems is a widely studied topic in discrete mathematics, e.g. \cite{minorsX,geelen,Diestel_Oum_2014,diestel2016tangles,diestel2016}.
On the other hand clustering, is an important and well studied topic in Data Science, e.g. \cite{carlsson2010characterization,dasgupta16hcoptimization,addad2017objective,kleinberg2003impossibility,luxburg2007spectral}.
Clustering is the umbrella term for different techniques to define sets of data points that are very similar to each other and not so similar to the data points contained in other sets.
We study the connection between both concepts by interpreting similarity as connectivity, thus two points are highly connected if their data is very similar and two sets are highly connected if they contain similar data points.
Both communities will benefit from such a connection, as it opens up a basis for a wide range of new results.
For example connectivity systems provide us with witnesses for the absence of highly connected regions, which is not yet established for clusters, as well as tree like representations of all those highly connected regions.
Additionally there is large variety of efficient algorithms to compute or approximate different kind of clusters, which can possibly be used to find algorithms for computing highly connected regions in connectivity systems.

Connectivity systems consist of a finite universe and a set function on subsets of the universe.
The set function is interpreted as the connectivity between a subset and its complement.
These so called connectivity functions are symmetric and submodular.
The concept is based on the notion of connectivity in graphs.
A prominent example is vertex connectivity.
There the universe is the set of all edges and the value of the vertex connectivity function is the number of vertices that are incident both to the subset and its complement.
In the context of connectivity systems two complementary questions are of interest \cite{grohe2016}:

\begin{enumerate}
	\item What are the highly connected regions of the universe?
	\item How can we decompose the universe along low order separations?
\end{enumerate}

Intuitively a highly connected region is a region of the universe such that no low order cut can cut through the region.
Thus for every low order separation of the universe there is one side of the separation that contains the region, or at least a large part of it.
A tangle is a concept that describes such highly connected regions in an indirect way via the low order separations.
For every low order separation the tangle describes on which side of the separation the large part of the region can be found.
Nevertheless, a separation may cut off small parts of the region.
In this way, for any single point it is not clearly defined whether it is part of the region or not.
On the other hand, the orientation has to be consistent, meaning that all orientations have to point towards the same region.

The second question is addressed by branch decompositions.
They consist of a ternary tree and a mapping of the elements to the leafs of the tree.
Then, the edges of the tree represent separations of the universe.
The width of such a decomposition is the largest value of any such separation.
Both concepts have been introduced on graphs by Robertson and Seymour \cite{minorsX}.
An overview of branch decompositions and tangles for integer-valued functions, as well as their connection can be found in a survey of Grohe \cite{grohe2016}, which can be translated to real-valued functions.
Branch decompositions and tangles address opposite questions, but they are dual for submodular connectivity functions.
There is a tangle of a certain value if and only if there is no branch decomposition of smaller value.
This duality result has first been shown on graphs in \cite{minorsX}.
There have been other set functions, that are not submodular, for which duality has been shown.
For example Adler et al. \cite{adler2007} have shown duality up to a constant factor for so called hypertangle number and hyperbranch width in hypergraphs.
Diestel and Oum \cite{Diestel_Oum_2014} developed a general duality theorem in combinatorial structures.

Besides tangles, there is another approach to identify highly similar (connected) regions, called clustering.
We consider hierarchical clustering algorithms as introduced by Carlson and Mémoli \cite{carlsson2010characterization}.
The basis of such an algorithm is a function that describes the distance between two sets.
Single linkage clustering for example considers the distance of two sets to be the smallest distance between any element from one set to any element from the other set.
The hierarchical clustering algorithm then merges the sets with the smallest distance, assigning the resulting partition this distance as its value.
We allow more than one merge at once, if the distances are equal.
The resulting sequence of partitions is represented by a so called dendogram.
The resulting sequence can also be represented by an ultrametric space.
In fact dendograms and ultrametrics are equivalent.
Carlson and Mémoli \cite{carlsson2010characterization} also showed that ultrametrics are stable under the single linkage algorithm.
A question arises: How do clustering and tangles relate?
A first approach towards this question was done by Diestel and Whittle \cite{diestel2016}, where they analyzed tangles in digital images as a way to describe the meaningful parts of the image.
Recently Elbracht et al. \cite{elbfioklekneirendteelux20} have successfully applied tangles to different clustering scenarios.

\subsection{Results}

To find a correspondence between tangles and clustering one of the main tasks is to find connectivity functions, or functions with similar properties, that represent the different clustering methods.
Our main result is, that we are able to specify a function that is correspondent to single linkage hierarchical clustering.
Recall that in single linkage the sets with the smallest distance between points belonging to different sets are chosen.
For an arbitrary metric $\dist\colon U\times U\rightarrow \mathbb{R}$, the corresponding set function is the \emph{\mindistftk} $\mind\colon 2^U\rightarrow \mathbb{R}$, defined by 
\begin{equation*}
\mind(X)=\max_{x\in X,y\notin X}\exp(-\dist(x,y)),
\end{equation*}
for all $X\in 2^U\setminus\{\emptyset,U\}$ and $\mind(\emptyset)=\mind(U)=0$.
We show that this function is not a classical connectivity function, since it is not submodular.
Therefore we introduce a new property, that we call maximum-submodularity.
A formal introduction of the \mindistftk\ and maximum-submodular functions is given in \cref{sec_mind}.

Our first theorem proves duality between branch decompositions and tangles of maximum-submodular functions.
\begin{theorem}
	Let $U$ be a finite set and let $\phi$ be a maximum-submodular connectivity function.
	The maximum order of a tangle of $\phi$ equals the minimum order of a branch decomposition of $\phi$.
	The existence of one is witness to the non-existence of the other.
	\label{theo:duality-informal}
\end{theorem}
A formal version of this theorem and its proof are shown in \cref{sec_dual}.
This duality is a key result in the theory of tangles of connectivity systems and suggests that the chosen property on the functions results in a similarly deep theory.
It allows us to use maximum-submodular connectivity functions to establish a connection between tangles and clustering.

Our second main result says that the tangles of the \mindistftk\ are in one-to-one correspondence with the resulting dendogram of single linkage hierarchical clustering.
For every non-singular set contained in a partition of the dendogram we find a distinct $\mind$-tangle and vice versa.
The technical notions appearing in the statement of the theorem will be explained later in this paper.
\begin{theorem}
	Let $(U,\dist)$ be a metric space.
	\begin{enumerate}
		\item For every $r\in\mathbb{R}$ and every cluster $\mathcal{B}$ of the dendogram resulting from single linkage with $|\mathcal{B}|>1$, \begin{equation*}
		\mathcal{T}\coloneqq\{X\subseteq U \mid \mind(X)<\exp(-r),~\mathcal{B}\subseteq X\}
		\end{equation*}
		is a $\mind$-tangle of $U$ of order $\exp(-r)$.
		\item For every $\mind$-tangle $\mathcal{T}$ of $U$ of order $k$ we can identify a cluster $\mathcal{B}$ of the dendogram resulting from single linkage with $|\mathcal{B}|>1$ such that
		\begin{equation*}
		\mathcal{T}=\{X\subseteq U \mid \mind(X)<k,~\mathcal{B}\subseteq X\}.
		\end{equation*}
	\end{enumerate}
	\label{theo:cluster-tangle-informal}
\end{theorem}

But we can go beyond single linkage clustering.
We find a one-to-one correspondence between dendograms and tangles of maximum-submodular connectivity functions, where the output is restricted to $[0,1)$\footnote{The restriction to outputs in $[0,1)$ stems from the choice of the order-reversing bijection $\exp(-r)$. Choosing a order reversing bijection $f\colon\mathbb{R}\rightarrow\mathbb{R}$ would give an equivalence for arbitrary maximum-submodular connectivity function.}.
This means that (the tangles of) maximum-submodular connectivity functions are equivalent to dendograms and thereby arbitrary hierarchical clustering algorithms.
The second statement of the theorem was first proposed by Nathan Bowler \cite{bowler}.
Again all technical notions appearing in the statement of the theorem will be explained later in this paper.
\begin{theorem}
	Let $U$ be some finite universe.
	\begin{enumerate}
		\item For every dendogram $\theta\colon[0,\infty)\rightarrow\mathcal{P}(U)$ there is a maximum-submodular connectivity function $\kappa_\theta\colon 2^U\rightarrow [0,1)$ such that, for every $r\in\mathbb{R}$ and every cluster $\mathcal{B}\in\theta(r)$ with $|\mathcal{B}|>1$, \begin{equation*}
		\mathcal{T}\coloneqq\{X\subseteq U \mid \kappa_\theta(X)<\exp(-r),~\mathcal{B}\subseteq X\}
		\end{equation*}
		is a $\kappa_\theta$-tangle of $U$ of order $\exp(-r)$.
		\item For every maximum-submodular connectivity function $\kappa\colon 2^U\rightarrow [0,1)$ there is a dendogram $\theta_\kappa\colon[0,\infty)\rightarrow\mathcal{P}(U)$ such that, for every $\kappa$-tangle $\mathcal{T}$ of $U$ of order $k$ we can identify a cluster $\mathcal{B}\in\theta_\kappa(-\ln(k))$ with $|\mathcal{B}|>1$ such that
		\begin{equation*}
		\mathcal{T}=\{X\subseteq U \mid \kappa(X)<k,~\mathcal{B}\subseteq X\}.
		\end{equation*}
	\end{enumerate}
	\label{theo:dendogram-tangle-informal}
\end{theorem}
An introduction to hierarchical clustering and the formal statements and proofs of \cref{theo:cluster-tangle-informal} and \cref{theo:dendogram-tangle-informal} can be found in \cref{sec_hc}.

This is to the best of our knowledge the first precise technical connection between tangles and clusters.
\footnote{A conference version that contains some of the results has appeared in \cite{fluck2019}.}

\section{Tangles and Branch Decompositions}

In our definitions we follow \cite{grohe2016}.
Our goal is to describe connectivity within some data set $U$.
Therefore we define set functions, that aim to describe how strong the connection is between a set and its complement.
For some finite universe \(U\) and any subset \(X\subseteq U\), we write \(\overline{X}\coloneqq U\setminus X\) for the complement.
We say a function $\kappa$ is \emph{normalized} if $\kappa(\emptyset)=0$, $\kappa$ is \emph{symmetric} if $\kappa(X)=\kappa(\overline{X})$, for all $X\subseteq U$ and $\kappa$ is \emph{submodular} if $\kappa(X)+\kappa(Y) \geq \kappa(X\cap Y)+\kappa(X\cup Y)$, for all $X,Y\subseteq U$.
A set function that is normalized, symmetric and submodular is called \emph{submodular connectivity function}.

\begin{example}[see \cite{grohe2016}]
	Let $G=(V,E)$ be a graph with vertex weights $w_V\colon V\rightarrow\mathbb{R}$.
	For any edge set \(X\subseteq E\) we write \(\Delta(X)\) to denote all vertices that are incident to both an edge in \(X\) and in  \(\overline{X}\).
	The weighted vertex-connectivity function $\nu\colon 2^E\rightarrow\mathbb{R}$, defined as
	\begin{equation*}
	\nu(X) \coloneqq \sum_{v \in \Delta(X)} w_V(v),
	\end{equation*}
	is a submodular connectivity function.
	\label{ex:vertex-con}
\end{example}

We introduce a different type of function, that also describes connectivity.
To show that this type has similar properties, we first take a look at basic concepts from the theory of connectivity systems.
Most of these concepts have only been studied for integer-valued functions, but for our needs all properties are translatable to real-valued functions.
We start with a formal definition of tangles, which are a way to describe highly connected regions.

\begin{definition}
	Let $\kappa$ be a symmetric set function on the universe $U$.
	A \emph{$\kappa$-tangle} of order $\ord(\mathcal{T})=k\geq0$ is a set $\mathcal{T}\subseteq 2^U$ such that:
	\begin{description}
		\item[\textbf{T.0}] $\kappa(X)<k$ for all $X\in \mathcal{T}$,
		\item[\textbf{T.1}] for all $X\subseteq U$ with $\kappa(X)<k$, either $X\in \mathcal{T}$ or $\overline{X}\in \mathcal{T}$ holds,
		\item[\textbf{T.2}] $X_1\cap X_2\cap X_3\neq\emptyset$ for all $X_1,X_2,X_3\in \mathcal{T}$ and
		\item[\textbf{T.3}] $\{x\}\notin\mathcal{T}$ for all $x\in U$.
	\end{description}
	\label{def:tangle}
\end{definition}

We define the \emph{tangle number} $\tn(\kappa)$ of a symmetric set function $\kappa$ to be the largest possible order for which we can still define a $\kappa$-tangle.

We use the following well-known lemma, which states that tangles are in a way closed under intersection and supersets.

\begin{lemma}[see \cite{grohe2016}]
	Let $\mathcal{T}$ be a $\kappa$-tangle of order $k$.
	Then it holds that 
	\begin{enumerate}
		\item for all $X\in\mathcal{T}$ and all $Y\supseteq X$, if $\kappa(Y)<k$ then $Y\in\mathcal{T}$ and
		\item for all $X,Y\in\mathcal{T}$, if $\kappa(X\cap Y)<k$ then $X\cap Y\in\mathcal{T}$.
	\end{enumerate}
	\label{lem:tangle-super-inter}
\end{lemma}

Before we continue to define an object dual to tangles, we need some notation from graph theory.
An undirected graph \(G=(V(G),E(G))\) is a pair of vertices \(V(G)\) and edges \(E(G)\subseteq V(G)\times V(G)\) where \(E(G)\) is symmetric, that is \( (u,v)\in E(G)\) if and only if \( (v,u)\in E(G)\).
We write \((u,v)\in E(G)\) to address both orientations of an edge simultaneously and if we want to address the orientations separately, we write \(\overrightarrow{E}(G)\) instead of \(E(G)\).
If \(G\) is clear from the context we write \(V,E\) instead of \(V(G),E(G)\).
Let \(v\in V\) be an arbitrary vertex.
We write \(N(v)\coloneqq\{u\mid (u,v)\in E\}\) to denote the neighborhood of \(v\).
An undirected tree is a graph \(T=(V,E)\) that does not contain cycles.
By \(L(T)\) we denote the set of all leaves of \(T\), that is \(L(T) \{ v\in V(T) \mid |N(v)|=1\}\).
A ternary tree is a tree \(T=(V,E)\) where all inner vertices have degree three, thus \(|N(V)|=3\), for all \(v\in V\setminus L(T)\).

A different way to describe connectivity in a universe is given by branch decompositions.
Here we do not look for highly connected regions, but ask ourselves how we can separate the universe into its single elements, using only separations of small value.
The following definition is based on branch decompositions as introduced in \cite{minorsX}, but we find it useful to introduce some generalizations, see also \cite{grohe2016}.

\begin{definition}
	Let $U$ be a finite set.
	\begin{itemize}
		\item A \emph{pre-decomposition} of $U$ is a pair $(T,\gamma)$ consisting of a ternary (undirected) tree $T=(E(T),V(T))$ and a mapping $\gamma\colon\overrightarrow{E}(T)\rightarrow2^U$, from the set of directed edges of \(T\) to subsets of \(U\), such that
		\begin{itemize}
			\item $\gamma(t,u)=\overline{\gamma(u,t)}$, for all $(t,u)\in\overrightarrow{E}(T)$, and
			\item $\gamma(s,u_1)\cup\gamma(s,u_2)\cup\gamma(s,u_3)=U$, for all internal nodes $s\in V(T)$ with $N(s)=\{u_1,u_2,u_3\}$.
		\end{itemize}
		\item For leaves $\ell\in L(T)$ with neighbor $N(\ell)=\{u\}$, we write $\gamma(\ell)$ instead of $\gamma(u,\ell)$.
		We call the $\gamma(\ell)$ \emph{atoms} and define $\At(T,\gamma)\coloneqq\{\gamma(\ell)\mid\ell\in L(T) \}$.
		\item A pre-decomposition is \emph{complete} if $|\gamma(\ell)|=1$, for all leaves $\ell\in L(T)$.
		\item A pre-decomposition is \emph{exact} at an internal node $t\in V(T)$ with $N(t)=\{u_1,u_2,u_3\}$ if all $\gamma(t,u_i)$ are mutually disjoint.
		\item A \emph{decomposition} is a pre-decomposition that is exact at all internal nodes.
		\item A \emph{branch decomposition} is a complete decomposition.
		\item Let $\kappa$ be a set function on $U$.
		The \emph{width} of a pre-decomposition $(T,\gamma)$ is
		\begin{equation*}
		\wid(T,\gamma)\coloneqq\max\{\kappa(\gamma(t,u))\mid(t,u)\in\overrightarrow{E}(T) \}.
		\end{equation*}
	\end{itemize}
	\label{def:undir-decomp}
\end{definition}

We define the \emph{branch width} $\bw(\kappa)$ of a symmetric set function $\kappa$ to be the smallest possible width $\wid(T,\gamma)$ of any branch decomposition $(T,\gamma)$ on $U$.
An example of pre-decompositions, exactness and completeness is shown in \cref{fig:decompositions}.
A decomposition can always be represented only by its atoms.
The mapping on every edge then contains the union of all atoms, that the edge points towards, that is \[\gamma(s,t)=\bigcup_{\substack{\ell\in L(T)\\\text{unique path from }\ell\text{ to }t\text{ does not contain }s}} \gamma(\ell). \]
For submodular connectivity functions it is known that every pre-decomposition can be transformed into a decomposition of the at most same width such that every atom of the decomposition is a subset of some atom of the pre-decomposition.
This implies that a complete pre-decomposition can be transformed into a branch decomposition, as one can easily prune leafs where the atoms are empty.

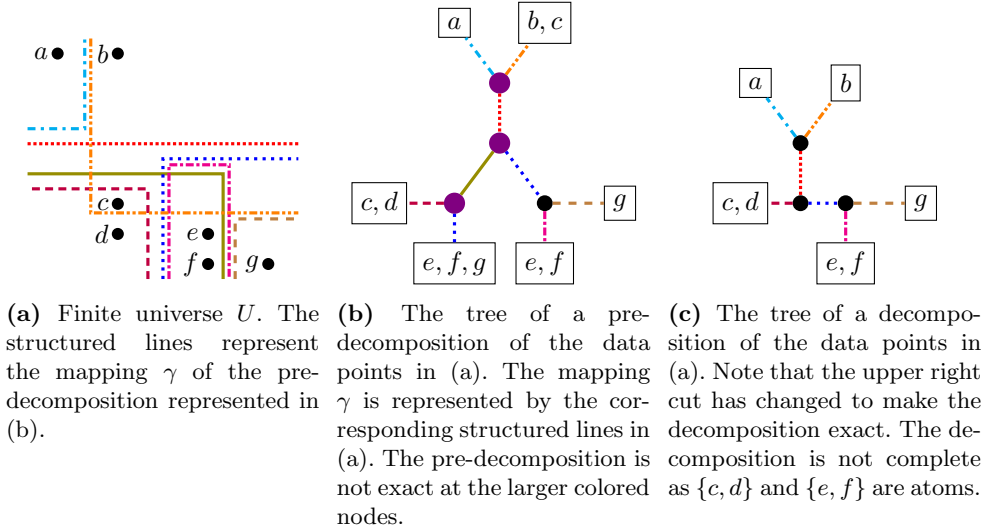
\begin{figure}[tbhp]
	\begin{subfigure}[t]{0.32\textwidth}
		\centering
		\begin{tikzpicture}[scale=0.4]
		\node (A) at (0,7) {};
		\fill (A) circle (0.2) node [left] {$a$};
		\node (B) at (2,7) {};
		\fill (B) circle (0.2) node [left] {$b$};
		\node (C) at (2,2) {};
		\fill (C) circle (0.2) node [left] {$c$};
		\node (D) at (2,1) {};
		\fill (D) circle (0.2) node [left] {$d$};
		\node (E) at (5,1) {};
		\fill (E) circle (0.2) node [left] {$e$};
		\node (F) at (5,0) {};
		\fill (F) circle (0.2) node [left] {$f$};
		\node (G) at (7,0) {};
		\fill (G) circle (0.2) node [left] {$g$};
		\path[draw, very thick, red, densely dotted] (-1,4) -- (8,4);
		\path[draw, very thick, olive] (-1,3) -- (5.5,3) -- (5.5,-0.5);
		\path[draw, very thick, blue, dotted] (3.5,-0.5) -- (3.5,3.5) -- (8,3.5);
		\path[draw, very thick, purple, densely dashed] (3,-0.5) -- (3,2.5) -- (-1,2.5);
		\path[draw, very thick, magenta, densely dashdotted] (3.7,-0.5) -- (3.7,3.3) -- (5.7,3.3) -- (5.7,-0.5);
		\path[draw, very thick, brown, dashed] (5.9,-0.5) -- (5.9,1.5) -- (8,1.5);
		\path[draw, very thick, cyan, dashdotted] (-1,4.5) -- (0.9,4.5) -- (0.9,7.5);
		\path[draw, very thick, orange, densely dashdotdotted] (1.1,7.5) -- (1.1,1.7) -- (8,1.7);
		
		\end{tikzpicture}
		\caption{Finite universe \(U\). The structured lines represent the mapping \(\gamma\) of the pre-decomposition represented in (b).}
		\label{fig:pre-dec-data}
	\end{subfigure}
	~
	\begin{subfigure}[t]{0.32\textwidth}
		\centering
		\begin{tikzpicture}[vertex/.style={fill,circle,inner sep=0pt, minimum size=6pt, outer sep=0pt},badvertex/.style={fill,circle,violet,inner sep=0pt, minimum size=8pt, outer sep=0pt},leaf/.style={draw,rectangle},scale=0.4]
		\node[badvertex] (a) at (0,0) {};
		\node[badvertex] (b) at (0,2) {};
		\node[vertex] (c) at (1.5,-2) {};
		\node[badvertex] (d) at (-1.5,-2) {};
		\node[leaf] (e) at (-4,-2) {$c,d$};
		\node[leaf] (f) at (-1.5,-4) {$e,f,g$};
		\node[leaf] (g) at (4,-2) {$g$};
		\node[leaf] (h) at (1.5,-4) {$e,f$};
		\node[leaf] (k) at (1.5,4) {$b,c$};
		\node[leaf] (l) at (-1.5,4) {$a$};
		\path[draw, very thick, red, densely dotted] (b) -- (a);
		\path[draw, very thick, blue, dotted] (a) -- (c);
		\path[draw, very thick, olive] (a) -- (d);
		\path[draw, very thick, purple, densely dashed] (d) -- (e);
		\path[draw, very thick, blue, dotted] (d) -- (f);
		\path[draw, very thick, brown, dashed] (c) -- (g);
		\path[draw, very thick, magenta, densely dashdotted] (c) -- (h);
		\path[draw, very thick, orange, densely dashdotdotted] (b) -- (k);
		\path[draw, very thick, cyan, dashdotted] (b) -- (l);
		\end{tikzpicture}
		\caption{The tree of a pre-decomposition of the data points in (a). The mapping \(\gamma\) is represented by the corresponding structured lines in (a). The pre-decomposition is not exact at the larger colored nodes.}
		\label{fig:pre-dec}
	\end{subfigure}
	~
	\begin{subfigure}[t]{0.32\textwidth}
		\centering
		\begin{tikzpicture}[vertex/.style={fill,circle,inner sep=0pt, minimum size=6pt, outer sep=0pt},leaf/.style={draw,rectangle},scale=0.4]
		\node[vertex] (a) at (0,0) {};
		\node[vertex] (b) at (0,2) {};
		\node[vertex] (c) at (1.5,0) {};
		\node[leaf] (e) at (-2,0) {$c,d$};
		\node[leaf] (g) at (4,0) {$g$};
		\node[leaf] (h) at (1.5,-2) {$e,f$};
		\node[leaf] (k) at (1.5,4) {$b$};
		\node[leaf] (l) at (-1.5,4) {$a$};
		\path[draw, very thick, red, densely dotted] (b) -- (a);
		\path[draw, very thick, blue, dotted] (a) -- (c);
		\path[draw, very thick, purple, densely dashed] (a) -- (e);
		\path[draw, very thick, brown, dashed] (c) -- (g);
		\path[draw, very thick, magenta, densely dashdotted] (c) -- (h);
		\path[draw, very thick, orange, densely dashdotdotted] (b) -- (k);
		\path[draw, very thick, cyan, dashdotted] (b) -- (l);
		\end{tikzpicture}
		\caption{The tree of a decomposition of the data points in (a). Note that the upper right cut has changed to make the decomposition exact. The decomposition is not complete as \(\{c,d\}\) and \(\{e,f\}\) are atoms.}
		\label{fig:dec}
	\end{subfigure}
	\caption{An example of a pre-decomposition and a decomposition, that is a pre-decomposition that is exact at every vertex. To transform the decomposition in (c) into a branch decomposition, one would need to add two children each to the leafs containing \(\{c,d\}\) and \(\{e,f\}\) separating the sets into singletons.}
	\label{fig:decompositions}
\end{figure}

For submodular connectivity functions duality between branch decompositions and tangles has been proven.
Robertson and Seymour \cite{minorsX} were the first to discover this duality for the unweighted vertex connectivity function as in \hyperref[ex:vertex-con]{Example~\ref*{ex:vertex-con}}\footnote{The definitions by Robertson and Seymour \cite{minorsX} are slightly different but equivalent up to some small issues with tangles of order \(\leq 2\) and isolated vertices or edges with at least one endvertex of degree 1, see \cite{grohe2016}.}.
Duality between branch decompositions and tangles states that a branch decomposition of a certain width is a witness for the non-existence of a tangle of any larger order and vice versa.

\begin{theorem}
	Let \(\kappa\) be a submodular connectivity function. It holds that 
	\begin{equation*}
	\tn(\kappa)=\bw(\kappa).
	\end{equation*}
	\label{theo:duality-submod_informal}
\end{theorem}

\section{The \MinDistFtk and Maximum-Submodularity}
\label{sec_mind}

In clustering applications the transformation $\exp(-c\cdot\operatorname{f}(x,y))$, for some constant $c$ and some function $\operatorname{f}$, is often used to transform a dissimilarity function $\operatorname{f}$ like a metric into a similarity function.
We use this transformation to define a connectivity function that captures the smallest distance between any point in a set and any point in the complement, that is $X$ is assigned a high value if there is a point outside of $X$ that is very close to a point in $X$.

\begin{definition}[\MinDistFtk]
	Let \(U\) be a finite data set and $\dist\colon U\times U\rightarrow \mathbb{R}$ be an arbitrary metric.
	The \emph{\mindistftk} $\mind\colon 2^U\rightarrow\mathbb{R}$, is defined as follows:
	\begin{equation*}
	\mind(X)\coloneqq\begin{cases}
	0 & \text{if } X=\emptyset \text{ or } \overline{X}=\emptyset,\\
	\max_{x\in X,x'\in\overline{X}}\exp(-\dist(x,x')) & \text{otherwise}.
	\end{cases}
	\end{equation*}
	\label{def_minDist}
\end{definition}

The \mindistftk\ is in general not submodular, as can be seen with a small example.
For an arbitrary $x\in\mathbb{R}^+$ define a one-dimensional universe containing only the following four points $a_1=x$, $a_2=x+1$, $a_3=-x$ and $a_4=-x-1$.
Let $X=\{a_1,a_2\}$, $Y=\{a_1,a_3\}$ and the metric $d(u,v)=|u-v|$ is the absolute of the difference.
Then $\mind(X)=\exp(-2x)<\mind(Y)=\mind(X\cap Y)=\mind(X\cup Y)=\exp(-1)$, for all $x>\frac{1}{2}$

We define a new property, that is similar to submodularity, which allows us to develop similar theories as for submodular connectivity functions.

\begin{definition}
	A set function $\kappa$ on a finite set $U$ is \emph{maximum-submodular} if, for all $X,Y\subseteq U$,
	\begin{equation*}
	\max(\kappa(X),\kappa(Y)) \geq \max(\kappa(X\cap Y),\kappa(X\cup Y)).
	\end{equation*}
\end{definition}

This property is neither a generalization of submodularity nor a specialization.
For instance $\nu$ as in \hyperref[ex:vertex-con]{Example~\ref*{ex:vertex-con}} is submodular but not maximum-submodular and in this section we see that the \mindistftk, which in general is not submodular, is maximum-submodular.
We call a normalized, symmetric and maximum-submodular set function \emph{maximum-submodular connectivity function}.

\begin{lemma}
	The \mindistftk\ is a maximum-submodular connectivity function.
\end{lemma}
\begin{proof}
	The \mindistftk\ is normalized by definition and symmetric since metrics are symmetric.
	If $X$ or $Y$ are equal to $\emptyset$ or $U$, maximum-submodularity trivially holds as $\{X,Y\}=\{X\cup Y,X\cap Y \}$ in these cases.
	If \(X\cap Y=\emptyset\) then we only need to show \(\max(\mind(X),\mind(Y))\geq\mind(X \cup Y)\).
	If \(X \cup Y= U\) this trivially holds, thus assume \(\overline{X}\cap \overline{Y}\neq\emptyset\) and choose $u\in \overline{X}\cap \overline{Y}$ and $v\in{X\cup Y}$ such that $\exp(-\dist(u,v))=\mind(X\cup Y)$.
	As either \(v\in X\) or \(v\in Y\) holds, it follows that \(\max(\mind(X),\mind(Y))\geq \exp(-\dist(u,v))=\mind(X\cup Y)\).
	Lastly if \(X\cup Y,\overline{X}\cup \overline{Y}\neq\emptyset\), we choose $u\in X\cap Y$ and $v\in\overline{X\cap Y}$ such that $\exp(-\dist(u,v))=\mind(X\cap Y)$.
	Analogously we choose $u'\in \overline{X \cup Y}=\overline{X}\cap\overline{Y}$ and $v'\in X\cup Y$.
	Then w.l.o.g. we distinguish four cases, depending on $v$ and $v'$.
	
	\textbf{Case 1:} $v\in \overline{X}\cap\overline{Y}$ and $v'\in X\cap Y$ hold: Then w.l.o.g. $v=u'$ and $u=v'$ hold.
	Therefore $\mind(X\cap Y) = \mind(X\cup Y)$ and thus $\mind(X)\geq\mind(X\cap Y)$ and $\mind(Y)\geq\mind(X\cup Y)$ hold.
	
	\textbf{Case 2:} $v\in \overline{X}\cap\overline{Y}$ and $v'\in \overline{X}\cap Y$ hold: It follows that $\mind(X)\geq\mind(X\cap Y)$ and $\mind(Y)\geq\mind(X\cup Y)$ hold.
	
	\textbf{Case 3:} $v,v'\in \overline{X}\cap Y$ holds: It follows that $\mind(X)\geq\mind(X\cap Y)$ and $\mind(Y)\geq\mind(X\cup Y)$ hold.
	
	\textbf{Case 4:} $v\in X\cap\overline{Y}$ and $v'\in \overline{X}\cap Y$ hold: In this case it holds that $\mind(Y)\geq\mind(X\cap Y)$ and $\mind(Y) \geq \mind(X\cup Y)$.
	Therefore we have $\mind(Y) \geq \max( \mind(X\cap Y) , \mind(X\cup Y) )$ and the inequality holds.
	
	As all other cases are symmetric to the four cases shown above, the inequality holds for all $X,Y\subseteq U$.
\end{proof}

Next, we consider tangles of the Minimum Distance Function.
Firstly, we give an example of such a tangle.

\begin{example}
	Let $U\subset\mathbb{R}^n$ be a finite set of points.
	Let $x_1,x_2\in U$ be two points such that $\dist(x_1,x_2)=\min\{\dist(x,y)\mid x,y\in U, x\neq y \}$.
	Then, for every $k\leq\exp(-\dist(x_1,x_2))$,
	\begin{equation*}
	\mathcal{T}\coloneqq\{X\subseteq U \mid \mind(X)<k,~x_1,x_2\in X\}
	\end{equation*}
	is a $\mind$-tangle of order $k$.
	
	$\mathcal{T}$ satisfies (T.0), (T.2) and (T.3) by construction.
	To see that (T.1) is satisfied note that if $x_1\in X$ and $x_2\in\overline{X}$ holds then $\mind(X)=\exp(-\dist(x_1,x_2))\geq k$ holds.
	\label{ex:tangle-maxd}
\end{example}

Having this example we realize that the tangles described are the only $\mind$-tangles. In fact all tangles of a maximum-submodular connectivity function behave similarly.

\begin{lemma}
	Let $\kappa$ be a maximum-submodular connectivity function. Every $\kappa$-tangle of order $k$ is of the form described as in \hyperref[ex:tangle-maxd]{Example~\ref*{ex:tangle-maxd}}.
	That is, we can identify two points $u,v\in U$ such that for all $X\in\mathcal{T}$ we have $u,v\in X$. Especially if $\kappa=\mind$ we have $\exp(-\dist(u,v))\geq k$.
	\label{lem:tangle-mind}
\end{lemma}
\begin{proof}
	To prove this we define the relation $\kappa^k\coloneqq\{(u,v)\mid \min_{\substack{X\subseteq U\\ u\in X, v\in\overline{X}}} \kappa(X)\geq k \}$ and consider the graph $G_k\coloneqq(U,\kappa^k)$.
	For every set $X\subseteq U$ with $\kappa(X)<k$ and every connected component $C$ of $G_k$, holds either $V(C)\subseteq X$ or $V(C)\subseteq\overline{X}$ by definition.
	Thus for every $\kappa$-tangle $\mathcal{T}$ of order $\leq k$ it holds that all $X\in\mathcal{T}$ are disjoint unions of connected components of $G_k$.
	Additionally, using \cref{lem:tangle-super-inter} (2) every such $\mathcal{T}$ is closed under intersection, as for two sets $X,Y\in\mathcal{T}$ and $k>\max\{\kappa(X),\kappa(Y)\}\geq \max\{\kappa(X\cap Y),\kappa(X\cup Y)\}\geq \kappa(X\cap Y) $.
	Suppose for contradiction there is a $\kappa$-tangle $\mathcal{T}$ of order $k$ such that there is no connected component of $G_k$, of size at least two, that is contained in all $X\in\mathcal{T}$.
	Let $C_0,\ldots,C_n$ be an enumeration of all connected components of $G_k$ with $|C_i|\geq2$.
	Then we can identify a sequence $X_0,\ldots,X_n\in\mathcal{T}$ such that $V(C_i)\nsubseteq X_i$, thus $V(C_i)\cap X_i=\emptyset$.
	We set $Y_1\coloneqq X_0\cap X_1$ and $Y_{i+1}\coloneqq Y_i\cap X_{i+1}$.
	Since $\mathcal{T}$ is closed under intersection, we get $Y_1,\ldots,Y_n\in\mathcal{T}$.
	As $\mathcal{T}$ is a tangle, $|Y_n|>1$ has to hold.
	For every subset $Y\subseteq Y_n$ we have $\kappa(Y)<k$ as \(Y_n\) is contains only isolated vertices of $G_k\coloneqq(U,\kappa^k)$.
	Take an enumeration of all elements $y_1,\ldots,y_m\in Y_n$ and construct a series of sets $Z_1,\ldots,Z_\ell\in\mathcal{T}$ such that $|Z_\ell|=1$.
	Clearly this contradicts the existence of $\mathcal{T}$.
	If $\{y_1\}\in\mathcal{T}$ set $Z_1\coloneqq\{y_1\}$, else set $Z_1\coloneqq Y_n \backslash \{y_1\}=Y_n \cap \overline{\{y_1\}} \in \mathcal{T}$.
	If $|Z_i|=1$ set $\ell=i$ and stop the construction.
	Otherwise if $\{y_{i+1}\}\in\mathcal{T}$ set $Z_{i+1}\coloneqq\{y_{i+1}\}$, else set $Z_{i+1}\coloneqq Z_i \backslash \{y_{i+1}\}=Z_i \cap \overline{\{y_{i+1}\}} \in \mathcal{T}$.
	As $|Z_i|>|Z_{i+1}|$ this construction terminates and yields the desired contradiction.
	
	Lastly if $\kappa=\mind$ we have $\mind^k=\{(u,v)\mid \exp(-\dist(u,v))\geq k \}$ as $\min_{\substack{X\subseteq U\\ u\in X, v\in\overline{X}}}=\exp(-\dist(u,v))$.
	To see this we observe that for all $u,v\in U$, with $u\neq v$, and for all $X\subset U$, with $u\in X,~ v\in \overline{X}$, it holds that $\mind(X)\geq\exp(-\dist(u,v))$ by the definition of $\mind$.
	On the other hand let $U_{u,v}\coloneqq\{v'\mid \dist(u,v')<d(u,v) \}$, then $\mind(U_{u,v})\leq \exp(-\dist(u,v)))$.
\end{proof}

From this lemma an important corollary follows.
In \cref{sec_hc} we use this to identify for each tangle a cluster resulting from single linkage hierarchical clustering.

\begin{corollary}
	Let $\kappa$ be a maximum-submodular connectivity function.
	Let $\mathcal{T}$ be a $\kappa$-tangle of order $k$ over the universe $U$.
	There is a unique connected component $C$ of the graph $G=(U,\kappa^k)$, with $\kappa^k\coloneqq\{(u,v)\mid \min_{\substack{X\subseteq U\\ u\in X, v\in\overline{X}}}\geq k \}$, such that $C\subseteq X$, for all $X\in\mathcal{T}$.
	\label{cor_tangles-con-components}
\end{corollary}
\begin{proof}
	We already showed that there exists some component $C$ such that $C\subseteq X$, for all $X\in\mathcal{T}$.
	Assume there is some component $C'\neq C$ such that $C'\subseteq X$, for all $X\in\mathcal{T}$.
	Then we have $C'\in\mathcal{T}$ and thus $C\subseteq C'$ which contradicts $C'\neq C$. 
\end{proof}

For most of the following applications of maximum-submodular functions it will be necessary to restrict oneself to functions which are strictly positive on every non-trivial set, that is $\kappa(X)>0$ for all $\emptyset\neq X \neq U$. We can do so as otherwise there is a canonical partition of $U$ into minimal subsets $\{V_1,\ldots,V_n\}$ with $\kappa(V_i)=0$, for $i\leq n$, such that $\kappa(X)=\max_{i\leq n}\kappa(X\cap V_i)$.

\section[Duality for Maximum-Submodular Functions]{Duality for Maximum-Submodular Functions\sectionmark{Duality}}
\sectionmark{Duality}
\label{sec_dual}

Now we prove duality for all maximum-submodular connectivity functions, thus also for the minimal distance function.
We achieve a result similar to the theory for submodular connectivity functions, first shown in \cite{minorsX}.
To formulate the Duality Theorem we first need a definition.

\begin{definition}
	Let $\kappa$ be a symmetric set function on $U$ and $\mathcal{A}\subseteq 2^U$.
	\begin{itemize}
		\item A pre-decomposition $(T,\gamma)$ is \emph{over} $\mathcal{A}$ if $At(T,\gamma)\subseteq\mathcal{A}$.
		\item A $\kappa$-tangle $\mathcal{T}$ \emph{avoids} $\mathcal{A}$ if $\mathcal{T}\cap\mathcal{A}=\emptyset$.
	\end{itemize}
	\label{def_over-avoid}
\end{definition}

The Duality Theorem states that there can not be any decomposition over a family of sets, if there is a tangle avoiding this family and vice versa.
The proof yields a construction of such a decomposition.
The following theorem is a precise formulation of \cref{theo:duality-informal}.

\begin{theorem}[Duality Theorem of Maximum-Submodular Functions]
	Let $\kappa$ be a maxi\-mum-submodular connectivity function on $U$.
	Let $\mathcal{A}\subseteq2^U$ such that $\mathcal{A}$ is closed under taking subsets and $\Sing(U)\subseteq\mathcal{A}$, where $\Sing(U)$ is the set of all singletons from $U$.
	Then there is a decomposition of width less than $k$ over $\mathcal{A}$ if and only if there is no $\kappa$-tangle of order $k$ that avoids $\mathcal{A}$.
	\label{theo:dual-sub-bound}
\end{theorem}

Assuming the theorem holds, we can directly derive the following corollary using that every branch decomposition of $U$ is complete, thus is over $\Sing(U)$ and every $\kappa$-tangle avoids $\Sing(U)$ by definition.

\begin{corollary}
	Let $\kappa$ be a maximum-submodular connectivity function on $U$.
	It holds that
	\begin{equation*}
	\tn(\kappa)=\bw(\kappa).
	\end{equation*}
	\label{cor_tangle-vs-bw-sub-bound}
\end{corollary}

Looking at the proof of duality for submodular connectivity functions as it is presented in \cite{grohe2016}, we see that they do not use any properties of the set function, besides symmetry and a transformation from a pre-decomposition into a decomposition of equal width.
Therefore, we can adapt that proof if we are able to do a similar transformation.
The following lemma shows how to achieve exactness at every node of a pre-decomposition.

\begin{lemma}[Exactness Lemma]
	Let $\kappa$ be a maximum-submodular connectivity function on $U$ and $(T,\gamma)$ be a pre-decomposition of $U$.
	Then there is a mapping $\gamma'\colon \overrightarrow{E}(T)\rightarrow 2^U$ such that $(T,\gamma')$ is a decomposition of $U$ satisfying
	\begin{itemize}
		\item $\wid(T,\gamma')\leq\wid(T,\gamma)$ and
		\item $\gamma'(\ell)\subseteq\gamma(\ell)$, for all leaves $\ell\in L(T)$.
	\end{itemize}
	\label{lem:dir-exactness}
\end{lemma}
\begin{proof}
	We iteratively construct $\gamma'$ from $\gamma$, keeping the invariants
	\begin{itemize}
		\item $\wid(T,\gamma')\leq\wid(T,\gamma)$,
		\item $\gamma'(s,t)\subseteq\gamma(s,t)$ or $\gamma'(s,t)\supseteq\gamma(s,t)$,for all edges \((s,t)\in\overrightarrow{E}(T)\), and
		\item $\gamma'(\ell)\subseteq\gamma(\ell)$, for all leaves $\ell\in L(T)$.
	\end{itemize}
	We pick an arbitrary leaf $\ell_{start}\in L(T)$ and set $\gamma'(\ell_{start},s)\coloneqq\gamma(\ell_{start},s)$ as well as $\gamma'(s,\ell_{start})\coloneqq\gamma(s,\ell_{start})$ for $s\in N(\ell_{start})$.
	If $T$ only consists of at most two nodes we are done, since $(T,\gamma')$ is already a decomposition.
	Otherwise, we traverse the tree with breadth-first search starting at $\ell_{start}$.
	If we reach a node $s\in V(T)\backslash L(T)$ with predecessor $t\in V(T)$, we do the following.
	Let $u_1,u_2\in N(s)$ be the successors of $s$ and define $X\coloneqq\gamma'(s,t)$ and $Y_i\coloneqq\gamma(s,u_i)$, for $i=1,2$.
	
	If $X\cap(Y_1\cup Y_2)\neq\emptyset$ we update $Y_i$ to $ Y_i\cap\overline{X}$, for $i=1,2$.
	This step is consistent with the invariants as $\kappa(Y_i\cap \overline{X})\leq\max(\kappa(Y_i\cap\overline{X}),\kappa(Y_i\cup\overline{X}))\leq\max(\kappa(Y_i),\kappa(X))$, where the second inequality holds due to maximum-submodularity, and $Y_i\cap\overline{X}\subseteq\gamma(s,u_i)$, for $i=1,2$.
	
	If $Y_1\cap Y_2\neq\emptyset$, update $Y_1$ to $Y_1\cap\overline{Y_2}$.
	This step is again consistent with the invariants as $\kappa(Y_1\cap \overline{Y_2})\leq\max(\kappa(Y_1\cap\overline{Y_2}),\kappa(Y_1\cup\overline{Y_2}))\leq\max(\kappa(Y_1),\kappa(Y_2))$, where the second inequality holds due to maximum-submodularity and symmetry, and $Y_1\cap\overline{Y_2}\subseteq\gamma(s,u_1)$.
	
	Set $\gamma'(s,u_i)\coloneqq Y_i$ and $\gamma'(u_i,s) \coloneqq \overline{\gamma'(s,u_i)}$, for $i=1,2$.
	After these steps we know that $\gamma'$ is exact at $s$ and we do not change $\gamma'$ for any predecessor of $s$.
	
	When we reach a leaf $\ell\in L(T)$, we do not change $\gamma'$ and continue with the next node in the breadth-first search.
	
	This construction yields the desired mapping. 
\end{proof}

The construction above may result in a tree with empty leafs.
But such a leaf can be easily removed by deleting it and its neighbor, connecting the resulting open edges.

Now, we are ready to prove the Duality Theorem for Maximum-Submodular Functions.

\begin{proof}[Proof of \cref{theo:dual-sub-bound}, see \cite{grohe2016}]
	For the forward direction, we let $(T,\gamma)$ be a decomposition of $U$ over $\mathcal{A}$ of width less than $k$.
	Suppose, for contradiction, $\mathcal{T}$ is a $\kappa$-tangle of order $k$ that avoids $\mathcal{A}$.
	We orient the edges $E(T)$ such that they point in the direction of the set that is contained in the tangle.
	Such a set always exists as the width is less then $k$.
	Thus, formally we orient $(s,t)\in E(T)$ towards $t$ if $\gamma(s,t) \in \mathcal{T}$, and towards $s$ if $\gamma(t,s) = \overline{\gamma(s,t)} \in \mathcal{T}$.
	As in every oriented tree, there is at least one node $t\in V(T)$ such that all edges incident to $t$ are oriented towards $t$.
	If $t\in L(T)$ then $\gamma(t)\in\mathcal{A}$ and $\gamma(t)\in\mathcal{T}$ which contradicts the assumption that $\mathcal{T}$ avoids $\mathcal{A}$.
	Thus, $t$ is an internal node with $N(t)=\{u_1,u_2,u_3\}$.
	But since all $\gamma(t,u_i)$ are mutually disjoint and all $\gamma(u_i,t)\in\mathcal{T}$ this contradicts (T.2) as $\gamma(t,u_1)\cup\gamma(t,u_2)\cup\gamma(t,u_3)=U$ and thus $\gamma(u_1,t)\cap\gamma(u_2,t)\cap\gamma(u_3,t)=\emptyset$.
	It follows that such a $\kappa$-tangle can not exist and the forward direction holds.
	
	For the backward direction assume there is no $\kappa$-tangle of order $k$ that avoids $\mathcal{A}$.
	We will construct a pre-decomposition $(T,\gamma)$ of $U$ over $\mathcal{A}$ of width less than $k$.
	Using the Exactness Lemma and since $\mathcal{A}$ is closed under taking subsets it follows that a decomposition of $U$ over $\mathcal{A}$ exists.
	
	We construct such a pre-decomposition $(T,\gamma)$ inductively on the number of sets $X\subseteq U$ with $\kappa(X)< k$ and neither $X\in\mathcal{A}$ nor $\overline{X}\in\mathcal{A}$.
	
	In the base case, for all $X\subseteq U$ with $\kappa(X)< k$, holds $X\in\mathcal{A}$ or $\overline{X}\in\mathcal{A}$.
	We define $\mathcal{Y}\coloneqq\{\overline{X}\mid X\in\mathcal{A} \text{ with } \kappa(X)<k \}$.
	We know that $\mathcal{Y}$ can not be a tangle, as we assumed that there is no tangle of order $k$.
	Since (T.0) and (T.1) hold by assumption on $\mathcal{A}$, either (T.2) or (T.3) have to be false.
	If $\mathcal{Y}$ violates (T.2) there are three sets $Y_1,Y_2,Y_3\in\mathcal{Y}$ such that $Y_1\cap Y_2\cap Y_3=\emptyset$.
	Then $\overline{Y_1},\overline{Y_2},\overline{Y_3}\in\mathcal{A}$ and $\overline{Y_1}\cup\overline{Y_2}\cup\overline{Y_3}=U$.
	We set $T\coloneqq(\{\ell_1,\ell_2,\ell_3,t\},\{(\ell_i,t)\mid i=1,2,3\})$, $\gamma(t,\ell_i)\coloneqq\overline{Y_i}\in\mathcal{A}$ and $\gamma(\ell_i,t)\coloneqq Y_i$.
	Then, $(T,\gamma)$ is a pre-decomposition of $U$ over $\mathcal{A}$.
	If $\mathcal{Y}$ violates (T.3) there is some $x\in U$ such that $\{x\}\in\mathcal{Y}$ and thus $\overline{\{x\}}\in\mathcal{A}$.
	Since $\Sing(U)\subseteq\mathcal{A}$ we have $\{x\}\in\mathcal{A}$.
	We take $T\coloneqq(\{s,t\},\{(s,t)\})$ and $\gamma(s,t)\coloneqq\{x\}$, $\gamma(t,s)\coloneqq\overline{\{x\}}$.
	Then $(T,\gamma)$ is a pre-decomposition of $U$ over $\mathcal{A}$.
	
	In the inductive step, we have some $X\subseteq U$ with $\kappa(X)< k$ and neither $X\in\mathcal{A}$ nor $\overline{X}\in\mathcal{A}$.
	We chose $X'$ such that $|X'|$ is minimal with respect to the conditions above.
	We set $\mathcal{A}^1\coloneqq\mathcal{A}\cup2^{X'}$ and $\mathcal{A}^2\coloneqq\mathcal{A}\cup2^{\overline{X'}}$.
	By the induction hypothesis there are pre-decompositions $(T^1,\gamma^1)$ over $\mathcal{A}^1$ and $(T^2,\gamma^2)$ over $\mathcal{A}^2$.
	If $\At(T_i,\gamma_i)\subseteq\mathcal{A}$ than $(T_i,\gamma_i)$ is a pre-decomposition over $\mathcal{A}$ and we are done.
	Otherwise, we can assume that $(T^1,\gamma^1)$ is a decomposition, due to the \hyperref[lem:dir-exactness]{Exactness Lemma~\ref*{lem:dir-exactness}}, thus the $\gamma^1(\ell)$ for all $\ell\in L(T^1)$ are unique.
	There is some $\ell^1\in L(T^1)$ such that $\gamma^1(\ell^1)\notin\mathcal{A}$.
	As the width of $(T^1,\gamma^1)$ is less  than $k$ and no true subset $X''\subset X'$ fulfills $\kappa(X'')< k$  and neither $X''\in\mathcal{A}$ nor $\overline{X''}\in\mathcal{A}$, we know that $\gamma^1(\ell^1)=X'$ and that it is the only leaf with this condition.
	We denote its neighbor by $s^1$.
	Let us now consider all $\ell^2_1,\ldots,\ell^2_m\in L(T^2)$ with $\gamma^2(\ell^2_i)\notin\mathcal{A}$.
	We know that, for all $\ell^2_i$, we have $\gamma^2(\ell_i^2)\subseteq \overline{X'}$. We consider all $s^2_i$ with $N(\ell^2_i)=\{s^2_i\}$.
	We modify $\gamma^2$ by setting $\gamma^2(s^2_i,\ell^2_i)\coloneqq\overline{X'}$ and $\gamma^2(\ell^2_i,s^2_i)\coloneqq X'$.
	The result will still be a pre-decomposition.
	Then, we construct a pre-decomposition $(T,\gamma)$ of $U$ over $\mathcal{A}$.
	We take $m$ disjoint copies $(T^1_i,\gamma^1_i)$ of $(T^1,\gamma^1)$.
	We define
	\begin{equation*}
	V(T)\coloneqq\bigcup_{1\leq i\leq m} V(T^1_i)\backslash\{\ell^1_i\} \cup V(T^2)\backslash\{\ell^2_1,\ldots,\ell^2_m\}
	\end{equation*}
	and take the union of all edge sets where $\ell^1_i$ is replaced by $s^2_i$ and $\ell^2_i$ is replaced by $s^1_i$.
	Then, we define $\gamma\colon E(T)\rightarrow2^U$ by
	\begin{equation*}
	\gamma(s,t)\coloneqq\begin{cases}
	X' & \text{if } (s,t)=(s^1_i,s^2_i) \text{ for some } 1\leq i\leq m,\\
	\overline{X'} & \text{if } (s,t)=(s^2_i,s^1_i) \text{ for some } 1\leq i\leq m,\\
	\gamma^1(s,t) & \text{if } s,t\in V(T^1_i) \text{ for some } 1\leq i\leq m,\\
	\gamma^2(s,t) & \text{if } s,t\in V(T^2).
	\end{cases}
	\end{equation*}
	Then, $(T,\gamma)$ is a pre-decomposition of $U$ over $\mathcal{A}$ of width less than $k$.
\end{proof}

\section{Maximum-Submodular Functions and Hierarchical Clustering}
\label{sec_hc}

To establish the connection between tangles and hierarchical clustering, we use Agglomerative Hierarchical Clustering via single linkage on dissimilarity inputs.
A dissimilarity input is an instance, where a small function value describes a large similarity between the points.
The result of a hierarchical clustering algorithm is a dendogram. For an arbitrary set $U$, $\mathcal{P}(U)$ denotes the \emph{set of all partitions} of $U$.

\begin{definition}[\cite{carlsson2010characterization}]
	A \emph{dendogram} over a finite set $U=\{x_1,\ldots,x_n\}$ is a function $\theta\colon[0,\infty)\rightarrow\mathcal{P}(U)$, satisfying the following conditions:
	\begin{enumerate}
		\item $\theta(0)=\{\{x_1\},\ldots,\{x_n\} \}$,
		\item there exists $t_0$ such that $\theta(t)=\{U\}$ for all $t\geq t_0$,
		\item if $r\leq s$ then $\theta(r)$ is a \emph{refinement} of $\theta(s)$, that is for every $\mathcal{B}\in\theta(r)$ there is some $\mathcal{B}'\in\theta(s)$ such that $\mathcal{B}\subseteq\mathcal{B}'$, and
		\item for all $r$ there exists $\epsilon>0$ such that $\theta(r)=\theta(t)$ for all $t\in[r,r+\epsilon]$.
	\end{enumerate}
\end{definition}

\begin{figure}[tbhp]
	\begin{subfigure}[t]{0.49\textwidth}
		\centering
		\begin{tikzpicture}[scale=0.5]
		\node (A) at (0,7) {};
		\fill (A) circle (0.2) node [left] {$a$};
		\node (B) at (2,7) {};
		\fill (B) circle (0.2) node [left] {$b$};
		\node (C) at (2,2) {};
		\fill (C) circle (0.2) node [left] {$c$};
		\node (D) at (2,1) {};
		\fill (D) circle (0.2) node [left] {$d$};
		\node (E) at (5,1) {};
		\fill (E) circle (0.2) node [left] {$e$};
		\node (F) at (5,0) {};
		\fill (F) circle (0.2) node [left] {$f$};
		\node (G) at (7,0) {};
		\fill (G) circle (0.2) node [left] {$g$};
		\end{tikzpicture}
		\caption{Data points used to compute the dendogram.}
		\label{fig_dend_univ}
	\end{subfigure}
	~
	\begin{subfigure}[t]{0.49\textwidth}
		\centering
		\begin{tikzpicture}[scale=0.5]
		\node (A) at (1,0) {};
		\node (A1) at (1,-1) {$a$};
		\fill (A) circle (0.2);
		\node (B) at (2,0) {};
		\node (B1) at (2,-1) {$b$};
		\fill (B) circle (0.2);
		\node (C) at (3,0) {};
		\node (C1) at (3,-1) {$c$};
		\fill (C) circle (0.2);
		\node (D) at (4,0) {};
		\node (D1) at (4,-1) {$d$};
		\fill (D) circle (0.2);
		\node (E) at (5,0) {};
		\node (E1) at (5,-1) {$e$};
		\fill (E) circle (0.2);
		\node (F) at (6,0) {};
		\node (F1) at (6,-1) {$f$};
		\fill (F) circle (0.2);
		\node (G) at (7,0) {};
		\node (G1) at (7,-1) {$g$};
		\fill (G) circle (0.2);
		
		\node (1) at (0,1) {$1$};
		\node (2) at (0,2) {$2$};
		\node (3) at (0,3) {$3$};
		\node (4) at (0,4) {$4$};
		\node (5) at (0,5) {$5$};
		
		\draw (A) -- (1,2) -- (2,2) -- (B);
		\draw (C) -- (3,1) -- (4,1) -- (D);
		\draw (E) -- (5,1) -- (6,1) -- (F);
		\draw (5.5,1) -- (5.5,2) -- (7,2) -- (G);
		\draw (3.5,1) -- (3.5,3) -- (6,3) -- (6,2);
		\draw (1.5,2) -- (1.5,5) -- (5,5) -- (5,3);
		\end{tikzpicture}
		\caption{The dendogram corresponding to the data points using single linkage.}
		\label{fig_dend_dend}
	\end{subfigure}
	\caption{An example of a dendogram. The distance function used is $\ell_{SL}(X,Y)\coloneqq\min_{x\in X,y\in Y}\|x-y\|$, where $\|\cdot\|$ is the Euclidean norm. This distance function is used in single linkage.}
	\label{fig_dend}
\end{figure}
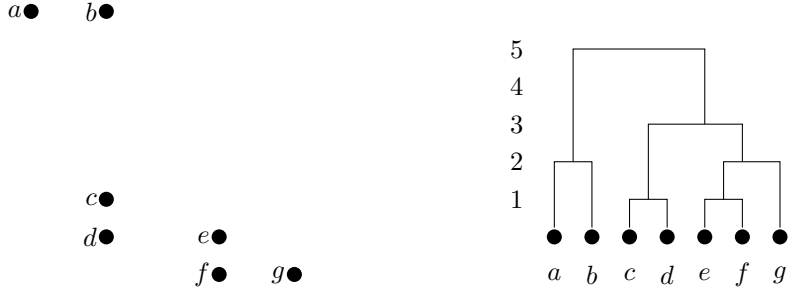

The first and second condition ensure that the trivial partitions are part of the dendogram, with the single elements having the smallest possible value and the whole set giving an upper bound.
The third condition states that every partition results from a merge of sets contained in a more refined partition.
The last condition is a bit technical, and ensures that $\theta$ is right continuous.
An example of a dendogram can be seen in \cref{fig_dend}.
We allow more than one cluster to merge in one step, as introduced and analyzed by Carlsson and Mémoli \cite{carlsson2010characterization}.
The single linkage clustering in this framework works as follows.

\begin{definition}[\cite{carlsson2010characterization}]
	Let $(U,\dist)$ be a metric space and let $\ell_{SL}\colon2^U\times2^U\rightarrow\mathbb{R}$ be the single linkage function on $U$ defined by
	\begin{equation*}
	\ell_{SL}(\mathcal{A},\mathcal{B})\coloneqq\min_{a\in\mathcal{A},b\in\mathcal{B}}\dist(a,b).
	\end{equation*}
	Inductively define a sequence of distances $R_0,R_1,\ldots,R_m\in[0,\infty)$ and a corresponding sequence of partitions $\Theta_0,\Theta_1,\ldots,\Theta_m\in\mathcal{P}(U)$ by:
	\begin{itemize}
		\item $R_0\coloneqq0$ and $\Theta_0\coloneqq\{\{x_1\},\ldots,\{x_n\}\}$, with $U=\{x_1,\ldots,x_n\}$,
		\item if $\Theta_i\neq\{U\}$, $R_{i+1}\coloneqq\min_{\substack{\mathcal{B},\mathcal{B'}\in\Theta_i\\\mathcal{B}\neq\mathcal{B'}}} \ell_{SL}(\mathcal{B},\mathcal{B'})$ and
		\item $\Theta_{i+1}\coloneqq\Theta_i/\sim_{R_{i+1}}$, where $\mathcal{B}\sim_{R_{i+1}}\mathcal{B'}$ if there exists a sequence of blocks of distance at most $R_{i+1}$, thus $\mathcal{B} = \mathcal{B}_1 ,\ldots, \mathcal{B}_s = \mathcal{B}'\in\Theta_i$ with $\ell_{SL}(\mathcal{B}_k,\mathcal{B}_{k+1})\leq R_{i+1}$, for $k=1,\ldots,s-1$ otherwise
		\item if $\Theta_i=\{U\}$, \(m\coloneqq i\).
	\end{itemize}
	Then the \emph{dendogram for single linkage} is defined by
	\begin{equation*}
	\theta^{\ell_{SL}}(r)\coloneqq\Theta_{i(r)},
	\end{equation*}
	where $i(r)\coloneqq\max\{i\mid R_i\leq r \}$.
	\label{def_sl-hc-clustering}
\end{definition}

A less technical way to describe this is, that we start with distance $R_0=0$ and the partition into single elements.
Then we inductively compute the smallest pairwise distance of any two points separated by the partition $\Theta_i$, store this as the next distance value $R_{i+1}$ and merge the corresponding sets to achieve a new partition $\Theta_{i+1}$.
We repeat this step until all sets are merged.

The result of an hierarchical clustering algorithm can also be represented in a different way, by so called ultrametrics.

\begin{definition}[Ultrametric]
	Let $(U,u)$ be a metric space. $(U,u)$ is a \emph{ultrametric space} if and only if the \emph{strong triangle inequality} holds, that is for all $x_1,x_2,x_3\in U$ it holds that
	\begin{equation*}
	\max(u(x_1,x_2),u(x_2,x_3))\geq u(x_1,x_3).
	\end{equation*}
\end{definition}

Note that the property for ultrametric spaces is a restriction of the triangle inequality for metric spaces. Indeed this property translates to all triangles being isoceles, where the third side is at most as long as the equal sides.
Carlsson and Mémoli \cite{carlsson2010characterization} showed that dendograms and ultrametrics are equivalent, and that is there is a natural bijection between the two.

\begin{theorem}[\cite{carlsson2010characterization}]
	Given a finite universe $U$, there is a bijection $\Psi\colon\mathcal{D}(U)\rightarrow\mathcal{U}(U)$ between the collection $\mathcal{D}(U)$ of all dendograms over $U$ and the collection $\mathcal{U}(U)$ of all ultrametrics over $U$ such that for any dendogram $\theta\in\mathcal{D}(U)$ the ultrametric $u=\Psi(\theta)$ generates the same hierarchical decomposition as $\theta$, that is,
	\begin{equation*}
	\text{for each } r\geq0,~x_1,x_2\in\mathcal{B}\in\theta(r) \Leftrightarrow u(x_1,x_2)\leq r.
	\end{equation*}
	Furthermore, this bijection is given by
	\begin{equation*}
	\Psi(\theta)(x,y) \coloneqq \min\{r\geq 0\mid x,y \text{ belong to the same block } \mathcal{B}\in\theta(r) \}.
	\end{equation*}
	\label{theo:dendogramm-ultrametric}
\end{theorem}

Additionally it can be shown that ultrametrics are stable under the single linkage algorithm, that is running the single linkage algorithm on any ultrametric results in the same ultrametric.
To see this we reformulate the single linkage algorithm in terms of the resulting ultrametric $\mathfrak{T}^{SL}$.
It is easy to see that every ultrametric is a fix point of $\mathfrak{T}^{SL}$, as the ultrametric property generalizes from triangles to any cycle.

\begin{lemma}[\cite{carlsson2010characterization}]
	Let $(U,\dist)$ be a finite metric space. The ultrametric $\mathfrak{T}^{SL}(U,\dist)=(U,u_{\dist})$ resulting from the single linkage hierarchical clustering algorithm is defined by
	\begin{equation*}
	u_{\dist}(x,y)\coloneqq \min_{\substack{k\in\mathbb{N}\\ x_0=x,x_1,\ldots,x_k=y\in U}}\left\{\max_{i<k} \dist(x_i,x_{i+1}) \right\}.
	\end{equation*}
	\label{lem:stable-ultrametric}
\end{lemma}

We need this fact later to proof the general equivalence between dendograms and tangles.
Towards this goal we first observe, that the resulting dendogram of the single linkage algorithm can be interpreted as a decomposition of the universe into its $\mind$-tangles, where the non-singular blocks of the dendogram correspond to the tangles.
The following theorem is a precise formulation of \cref {theo:cluster-tangle-informal}.

\begin{theorem}
	Let $(U,\dist)$ be a metric space.
	\begin{enumerate}
		\item For every $r\in\mathbb{R}$ and every block $\mathcal{B}\in\theta^{\ell_{SL}}(r)$ with $|\mathcal{B}|>1$, \begin{equation*}
		\mathcal{T}\coloneqq\{X\subseteq U \mid \mind(X)<\exp(-r),~\mathcal{B}\subseteq X\}
		\end{equation*}
		is a $\mind$-tangle of $U$ of order $\exp(-r)$.
		\item For every $\mind$-tangle $\mathcal{T}$ of $U$ of order $k$ there is a block $\mathcal{B}\in\theta^{\ell_{SL}}(-\log(k))$ with $|\mathcal{B}|>1$ such that
		\begin{equation*}
		\mathcal{T}=\{X\subseteq U \mid \mind(X)<k,~\mathcal{B}\subseteq X\}.
		\end{equation*}
	\end{enumerate}
	\label{theo:cluster-tangle}
\end{theorem}

\begin{proof}
	Using the same arguments as in \cref{ex:tangle-maxd} the first statement holds.
	For the second statement one needs the equivalence relation $\sim_r$ on $U$, where $x\sim_r y$ if and only if there is a sequence of elements $x=x_1,\ldots,x_s=y\in U$ such that $\dist(x_i,x_{i+1})\leq r$.
	Carlsson and Mémoli \cite{carlsson2010characterization} have shown that the blocks of $\theta^{\ell_{SL}}(r)$ are exactly the equivalence classes $U/{\sim_r}$.
	
	Let $\mathcal{T}$ be a $\mind$-tangle of order $k$.
	Using \cref{cor_tangles-con-components} we can find a connected component $C$ in the graph $G=(U,\mind^k)$ such that $C\subseteq X$, for all $X\in\mathcal{T}$.
	Looking at the definition $\mind^k\coloneqq\{(u,v)\mid \exp(-\dist(u,v))\geq k \}$ we see that two elements $u,v\in U$ are connected in $G$ if and only if for their distance holds $\dist(u,v)\leq-\log(k)$, thus $u\sim_{-\log(k)}v$.
	It follows that the equivalence classes of $U/{\sim_{-\log(k)}}$ are exactly the same as the connected components of $G$.
	Thus, there is a block $C=\mathcal{B}\in\theta^{\ell_{SL}}(-\log(k))$ that fulfills the requirement. 
\end{proof}

As ultrametrics are stable under single linkage, one sees that any dendogram can be represented by the tangles of some maximum-submodular connectivity function. Interestingly the connection is even stronger. For every maximum-submodular function connectivity function that has been scaled to range $[0,1)$ and if all non-trivial sets are strictly positive\footnote{The restriction of the range is due to the choice of the order-reversing bijection. All non-trivial sets have to be strictly positive as dendograms need to join all data points at some finite value, therefore the partition of these data points most have a non-zero value.} there is a corresponding dendogram.

\begin{definition}
	Let $U$ be some finite universe, $\theta$ some dendogram over $U$ and $\kappa$ some maximum-submodular connectivity function over $U$. $\theta$ and $\kappa$ are \emph{\equivalent} if:
	\begin{enumerate}
		\item For every $r\in\mathbb{R}$ and every block $\mathcal{B}\in\theta(r)$ with $|\mathcal{B}|>1$, \begin{equation*}
		\mathcal{T}\coloneqq\{X\subseteq U \mid \kappa(X)<\exp(-r),~\mathcal{B}\subseteq X\}
		\end{equation*}
		is a $\kappa$-tangle of $U$ of order $\exp(-r)$.
		\item For every $\kappa$-tangle $\mathcal{T}$ of $U$ of order $k$ we can identify a block $\mathcal{B}\in\theta(-\log(k))$ with $|\mathcal{B}|>1$ such that
		\begin{equation*}
		\mathcal{T}=\{X\subseteq U \mid \kappa(X)<k,~\mathcal{B}\subseteq X\}.
		\end{equation*}
	\end{enumerate}
\end{definition}

Using this definition we are ready to reformulate \cref{theo:dendogram-tangle-informal}. The second statement has first been shown by Nathan Bowler \cite{bowler}.

\begin{theorem}
	Let $U$ be some finite universe.
	\begin{enumerate}
		\item For every dendogram $\theta$ over $U$ there is an \equivalent~ maximum-submodular connectivity function $\kappa_\theta$.
		\item For every maximum-submodular connectivity function $\kappa$ over $U$ with range $[0,1)$ and if, for all $\emptyset\neq X\neq U$, $\kappa(X)>0$, then there is an \equivalent\ dendogram $\theta_\kappa$.
	\end{enumerate}
	\label{theo:dendogram-tangle}
\end{theorem}
\begin{proof}
	~
	\begin{enumerate}
		\item Let \(\theta\) be a dendogramm over some finite universe \(U\) and let \(u\coloneqq\Psi(\theta)\) be the ultrametric as in \cref{theo:dendogramm-ultrametric}. Using \cref{lem:stable-ultrametric} we see that applying the single linkage hierarchical clustering algorithm to \((U,u)\) results in the same ultrametric space, that is $\mathfrak{T}^{SL}(U,u)=(U,u)$, and therefore also \(\theta^{\ell_{SL}}=\theta\). Thus the statement follows by applying \cref{theo:cluster-tangle} to the (ultra-)metric space \((U,u)\).
		\item We show this statement by proving that there exists an ultrametric $u_\kappa$ such that $\minsome_{u_\kappa} = \kappa$. Then the statement follows from \cref{theo:cluster-tangle}.
		W.l.o.g. we have $\kappa(X)>0$, for all $\emptyset\neq X \neq U$. We define a function $u_\kappa\colon U^2\rightarrow \mathbb{R}^+$ as follows:
		\begin{equation*}
		u_\kappa(x,y) \coloneqq \begin{cases}
		\max_{\substack{X\subseteq U\\x\in X\\ y\in\overline{X}}} -\log(\kappa(X)) 
		& \text{if } x \neq y,\\
		0 & \text{else}.
		\end{cases}
		\end{equation*}
		\begin{claim}
			$u_\kappa$ is an ultra metric.
			\label{cl:ultrametric}
		\end{claim}
		\begin{claimproof}
			$u_\kappa$ is positive (as all values of $\kappa$ are in $[0,1)$), symmetric and the identity of indescernibles holds by definition.
			It remains to show, that for all $x,y,z\in U$ the strong triangle inequality holds.
			If $x=z$ the strong triangle inequality holds due to symmetry.
			Thus assume $x\neq z$.
			Let $X\subset U$ be chosen to minimize $\kappa(X)$ with restriction to $x\in X$ and $z\in\overline{X}$.
			Then we have $u_\kappa(x,z)=-\log(\kappa(X))$.
			First assume $y\in X$, then $u_\kappa(y,z)\geq-\log(\kappa(X))=u_\kappa(x,z)$.
			If on the other hand $y\in\overline{X}$, then we get $u_\kappa(x,y)\geq-\log(\kappa(X))=u_\kappa(x,z)$.
			In both cases we get $\max(u_\kappa(x,y),u_\kappa(y,z))\geq u_\kappa(x,z)$, thus $u_\kappa$ is an ultrametric.
		\end{claimproof}
		\begin{claim}
			$\minsome_{u_\kappa} = \kappa$.
			\label{cl:kappa}
		\end{claim}
		\begin{claimproof}
			First we show $\minsome_{u_\kappa}(X) \leq \kappa(X)$, for all $X\subseteq U$.
			Let $X\subseteq U$ be an arbitrary set.
			If $X = \emptyset$ or $X = U$, then $\minsome_{u_\kappa}(X)=0=\kappa(X)$.
			Otherwise let $x\in X$ and $y\in \overline{X}$.
			By definition $u_\kappa(x,y)\geq -\ln(\kappa(X))$ and therefore $\exp(-u_\kappa(x,y)) \leq \kappa(X)$.
			Thus $\minsome_{u_\kappa}(X)\leq \kappa(X)$.
			
			Next we show that $\minsome_{u_\kappa}(X) \geq \kappa(X)$, for all $X\subseteq U$, which concludes the proof.
			If $X\in\{\emptyset, U\}$, then $\kappa(X)=0\leq\minsome_{u_\kappa}$.
			Otherwise let $k\coloneqq\minsome_{u_\kappa}(X)$.
			For any $x\in X$, $y\in \overline{X}$ we have $k=\minsome_{u_\kappa}(X)\geq \exp(-u_\kappa(x,y))$ by definition of $\minsome_{u_\kappa}$, thus $u_\kappa(x,y) \geq -\log(k)$.
			By definition of $u_\kappa$, there is a set $X_{x,y}$ with $x\in X_{x,y}$ and $y\in\overline{X_{x,y}}$, such that $-\log(\kappa(X_{x,y})) \geq -\log(k)$, thus $\kappa(X_{x,y})\leq k$.
			For all $x\in X$ we define $X_x\coloneqq\bigcap_{y\in\overline{X}} X_{x,y}$.
			We have $\kappa(X_x)\leq k$ by the maximum-submodularity of $\kappa$.
			It is easy to see that $X_x\cap \overline{X}=\emptyset$.
			Thus we get $X=\bigcup_{x\in X}X_x$ and again by the maximum-submodularity of $\kappa$ we get $\kappa(X)\leq k = \minsome_{u_\kappa}(X)$.
		\end{claimproof}
		
		Combining \cref{cl:ultrametric} and \cref{cl:kappa} proofs the lemma.
	\end{enumerate}
\end{proof}

From \cref{theo:dendogram-tangle} follows that for every hierarchical clustering algorithm there exists some maximum-submodular connectivity function such that the tangles describe the resulting clusters.
But in contrast to the single linkage algorithm, as seen in \cref{theo:cluster-tangle}, we do not get a constructive result for arbitrary algorithms, that is we can only compute the connectivity function from the clustering results and not directly from the data points or the corresponding metric space.

\begin{remark}
	Let us consider two popular hierarchical clustering algorithms, average linkage and complete linkage.
	The algorithm is the same as in \cref{def_sl-hc-clustering}, but the linkage function $\ell$ changes.
	The distance of two sets for complete linkage equals the maximum distance of any point from one set to any point from the other set, thus $\ell_{CL}(X,Y)\coloneqq\max_{x\in X,y\in Y}\dist(x,y)$.
	Using the same trick as for single linkage, a natural related connectivity function is $\kappa_{\dist}(X)\coloneqq\min_{x\in X,y\in\overline{X}}\exp(-\dist(x,y))$, for $X\in2^U\backslash\{\emptyset,U\}$, and $\kappa_{\dist}=0$, otherwise. This function is maximum-submodular and using the Manhattan distance it is even submodular.
	But in general, for an arbitrary partition $P$, it holds that
	\begin{equation*}
	\min_{X,Y\in P}\max_{x\in_X,y\in Y} \dist(x,y)\neq-\log(\max_{X\in P}\min_{x\in X,x'\in\overline{X}}\exp(-\dist(x,x'))),
	\end{equation*}
	as $\ell_{CL}(X,\overline{X})=\max_{Y\in P}\ell_{CL}(X,Y)$, for arbitrary $X\in P$.
	This is different for single linkage.
	It holds that for any partition $P$ of the universe we have
	\begin{equation*}
	\min_{X,Y\in P}\ell_{SL}(X,Y)=-\log(\max_{X\in P}\mind(X)),
	\end{equation*}
	as $\ell_{SL}(X,\overline{X})=\min_{Y\in P}\ell_{SL}(X,Y)$, for arbitrary $X\in P$.
	Thus in contrast to single linkage, the optimum of complete linkage $\ell_{CL}$ used to compute $R_{i+1}$ does not correspond to the optimum according to the connectivity function $\kappa_{\dist}$.
	For average linkage ($\ell_{AL}(X,Y)\coloneqq\sum_{x\in X,y\in Y}\frac{\dist(x,y)}{|X||Y|}$), a corresponding set function could be $\varphi_{\dist}(X)\coloneqq\sum_{x\in X,y\in\overline{X}}\frac{\exp(-\dist(x,y))}{|X||\overline{X}|}$, for $X\in 2^U\backslash\{\emptyset,U \}$, and $\varphi_{\dist}=0$, otherwise.
	It is neither submodular nor maximum-submodular.
	Additionally in general $\ell_{AL}(X,\overline{X})$ is not directly computable from $\ell_{AL}(X,Y)$, for $X,Y\in P$ with $P$ an arbitrary partition.
	To compute $\ell_{AL}(X,\overline{X})$, also the size of all $Y\in P$ is needed.
	We have
	\begin{equation*}
	\ell_{AL}(X,\overline{X})=\frac{\sum_{Y\in P,X\neq Y}|Y|\ell_{AL}(X,Y)}{\sum_{Y\in P,X\neq Y} |Y|}.
	\end{equation*}
	This makes it even harder to find a suitable connectivity function.
\end{remark}

\section{Conclusion}

We establish a precise technical equivalence between tangles and hierarchical clustering.
We show that the result of any hierarchical clustering algorithm is in one-to-one correspondence with the tangles of a maximum-submodular connectivity function.
We can specify this connection for the \mindistftk\ and single linkage clustering, where we are able to define the equivalent connectivity function directly from the metric space of the data points.
It is still an open question if for other clustering algorithms we can find corresponding set functions directly from the metric space of the data.
One of the main obstacles here is, that tangles and the corresponding set functions only look at global connectivity of some set to its converse whereas hierarchical clustering looks at local connectivity between two sets.
For single linkage these two notions turned out to be the same.

Our second contribution is to show duality between tangles and branch decompositions for a new class of functions.
In our view, the key transformation in the proof of the Duality Theorem is the Exactness Lemma (related to „shifting“ in \cite{Diestel_Oum_2014}); this is where submodularity or maximum-submodularity comes in.
To broaden the theory, it will be essential to understand sufficient and necessary conditions such that these transformations are possible.

\bibliography{literatur}
\bibliographystyle{siamplain}

\appendix

\end{document}